\documentclass[aps,prl,twoside,twocolumn,notitlepage,superscriptaddress,10pt]{revtex4-1}
\usepackage[utf8]{inputenc}
\usepackage[T1]{fontenc}
\usepackage{graphicx}
\usepackage{mathrsfs}
\usepackage{hyperref}
\usepackage[english]{babel}
\usepackage{verbatim}
\usepackage{amssymb}
\usepackage{amsmath}
\usepackage[english]{babel}
\usepackage{amsthm}
\usepackage{mathtools}
\usepackage{amsfonts,bbm}
\usepackage{amssymb,amscd}
\usepackage{enumitem}
\usepackage{tikz}
\newtheorem{theorem}{Theorem}
\newtheorem{hypothesis}{Hypothesis}

\newcommand{\R}{\mathbb{R}}
\newcommand{\e}{\mathrm{e}}
\newcommand{\Hilb}{\mathcal{H}}

\DeclareMathOperator{\Tr}{Tr}

\newcommand{\1}{\mathbbm{1}}

\newcommand{\Ket}[1]{ \left| #1 \right\rangle}
\newcommand{\Bra}[1]{ \left\langle #1 \right|}
\newcommand{\Scal}[2]{ \left\langle #1 | #2 \right\rangle}

\newcommand{\MV}[1]{\left\langle #1 \right\rangle}

\newcommand{\braket}[2]{\langle #1 | #2 \rangle}
\newcommand{\ketbra}[2]{| #2 \rangle\langle #1 |}
\newcommand{\ket}[1]{| #1 \rangle}
\newcommand{\bra}[1]{\langle #1 |}
\renewcommand{\equiv}{\coloneqq}

\begin{document}
\date{\today}
\title{Degenerate observables and the many Eigenstate Thermalization Hypotheses}%
\author{Fabio Anza}
\affiliation{Clarendon  Laboratory,  University  of  Oxford,  Parks  Road,  Oxford  OX1  3PU,  United  Kingdom}
\affiliation{The Abdus Salam International Centre for Theoretical Physics (ICTP), Trieste, Italy}
\author{Christian Gogolin}
\affiliation{ICFO-Institut de Ciencies Fotoniques, The Barcelona Institute of Science and Technology, 08860 Castelldefels (Barcelona), Spain}
\author{Marcus Huber}
\affiliation{Institute for Quantum Optics and Quantum Information (IQOQI), Austrian Academy of Sciences, Boltzmanngasse 3, A-1090 Vienna, Austria}

\begin{abstract}
Under unitary time evolution, expectation values of physically reasonable observables often evolve towards the predictions of equilibrium statistical mechanics.
The eigenstate thermalization hypothesis (ETH) states that this is also true already for individual energy eigenstates. 
Here we aim at elucidating the emergence of ETH for observables that can realistically be measured due to their high degeneracy, such as local, extensive or macroscopic observables.
We bisect this problem into two parts, a condition on the relative overlaps and one on the relative phases between the eigenbases of the observable and Hamiltonian.
We show that the relative overlaps are completely unbiased for highly degenerate observables and demonstrate that unless relative phases conspire to cumulative effects this makes such observables verify ETH.
Through connecting the degeneracy of observables and entanglement of the energy eigenstates this result elucidates potential pathways to equilibration in a fully general way.
\end{abstract}

\maketitle 

``Pure state quantum statistical mechanics''\cite{Pure1,Cazalilla2010a,Pure1b,Pure2,Pure3} aims at understanding under which conditions the use of tools from statistical mechanics can be justified based on the first principles of standard quantum mechanics with as few extra assumptions as possible.
To explain the emergence of thermalization it combines three approaches:
Typicality arguments \cite{Neumann1929,Tumulka2010,Typ1,Typ2,Typ3,Typ4,Typ5,Typ6}, the dynamical equilibration approach \cite{Rei1,Rei2,tsc1,tsc2,tsc3,tsc4,tsc5,tsc6} and the Eigenstate Thermalization Hypothesis (ETH) \cite{Berry,Shnir,ETH0,ETH1,ETH2,ETH3,ETH4,ETH5,ETH6,ETH7,ETH8,ETH9,ETH10}. According to the first one, systems appear to be in equilibrium because, in a precise sense, most states are in equilibrium.
Alternatively, according to the second approach apparent equilibration of observables and whole subsystems emerges because initial states of large many-body systems overlap with many energy eigenstates and therefore explore a large part of Hilbert space during their evolution, almost all the while being almost indistinguishable from a static equilibrium state.
ETH, on the other hand, is a hypothesis about properties of individual eigenstates of sufficiently complicated quantum many-body systems which was suggested by various results in quantum chaos theory and it adduces the appearance of thermalization during such equilibration to an underlying chaotic behavior.
The basic idea is that, for large system sizes and in sufficiently complicated quantum many-body systems, the energy eigenstates can be so entangled that when we look at their overlaps with the basis of a physical observable they can be effectively described by random variables.
If the ETH is fulfilled, it guarantees thermalization whenever equilibration happens because of the mechanisms described above.
Depending on how broad one wants the class of initial states that thermalize to be, the fulfillment of the ETH is also a necessary criterion for thermalization \cite{Pure1b,DePalma2015}.

The ETH is sometimes criticized for its lack of predictive power, as it leaves open at least three important questions: what precisely are ``physical observables''; what makes a system ``sufficiently complicated'' to expect that ETH applies; how long will it take for such observables to reach thermal expectation values \cite{tsc1,tsc2}.
For this reason, a lot of effort has been focused on numerical investigations that validate the ETH in specific Hamiltonian models and for various observables, often including local ones.
The ETH is generally found to hold in non-integrable systems that are not many-body localized and equilibration towards thermal expectation values usually happens on reasonable times scales \cite{tsc1,tsc3,tsc2,ETH0}.

Recently \cite{MaxShan} it has been shown that for any Hamiltonian there is always a large number of observables which satisfy ETH.
They have been dubbed ``Hamiltonian Unbiased Observables'' (HUO) and admit an algorithmic construction.
Unfortunately this still leaves open when concrete physically relevant observables satisfy the ETH.
In this letter we make progress in this direction.
Building on the connection between HUOs and ETH, we present a theorem which can be used as a tool to investigate the emergence of ETH.
In order to show how it can be used, we present three applications:
local observables, extensive observables, and macro-observables.
We will give precise definitions for each of them later.

The paper is organized as follows.
First we set-up the notation, recall different formulations of the ETH and clarify which one we will be using throughout the paper.
We continue with a brief digression on physical observables and degeneracies and recall the concepts of Hamiltonian unbiased basis and observables.
We then present our main result, which elucidates the question under which conditions highly degenerate observables are HUO and discuss consequences of it for local observables, extensive observables, and a certain type of macro-observables.

\paragraph{Versions of the ETH.}
We start by reviewing several versions of the ETH that have appeared in the literature.
All versions of the ETH are statements about properties of large systems.
In principle one would hence state the following in terms of families of systems of increasing size/particle number.
To not over-complicate things we do not make this explicit and instead implicitly assume that a limit of large system size exists and makes sense and that it is understood that the following are meant as statements about asymptotic scaling.   
Throughout the paper we assume all Hamiltonians $H$ to be non-degenerate with eigenvalues $E_m$ and eigenstates $\ket{E_m}$.
For any given initial state of the form $\ket{\psi_0} = \sum_{m}c_m \ket{E_m}$ with $c_m \coloneqq \braket{E_m}{\psi_0}$ we denote by $\rho_{\mathrm{DE}} \coloneqq \sum_{m}|c_m|^2 \ket{E_m}\bra{E_m}$ the diagonal ensemble, also known as the time averaged state.

Before we continue, we review some variants of the eigenstate thermalization hypothesis. These are essentially different mathematical statements which aim at formalizing the same physical intuition. Our goal here is to provide a reasonable clusterization of the most used versions of ETH and to state which one we will refer to throughout the paper.
\begin{hypothesis}[Original ETH] \label{hyp:originaleth}
  The matrix elements $A_{m,n} \coloneqq \bra{E_m}A\ket{E_n}$ of any physically reasonable observable $A$ with respect to the energy eigenstates $\ket{E_m}$ in the bulk of the spectrum of a Hamiltonian of a system with $N$ particles satisfy
  \begin{subequations}
  \begin{align}
    &&- \ln | A_{m+1,m+1} - A_{m,m} | &\in \mathcal{O}(N) \\
    &\text{and} & - \ln |A_{m,n}| &\in \mathcal{O}(N) \,.
  \end{align}
  \end{subequations}
\end{hypothesis}
In words: Off-diagonal elements of physically reasonable observables and the differences between neighboring diagonal elements are exponentially small in the size of the system.
This kind of ETH is what Srednicki argued to be fulfilled in a hard-sphere gas \cite{ETH3}. Similar variants appeared for example in \cite{ETH5,ETH6,ETH7,Khemani2014,ETH0}.

\begin{hypothesis}[Thermal ETH] \label{hyp:thermaleth}
  There exists a function $\beta:\R\to\R_0^+$ such that for any physically reasonable observable $A$ the expectation values $A_{m} \coloneqq \bra{E_m}A\ket{E_m}$ of $A$ with respect to the energy eigenstates $\ket{E_m}$ in the bulk of the spectrum of a Hamiltonian of system are close to thermal in the sense that 
  \begin{equation}
    |A_{m} - \Tr(A\,\e^{-\beta(E_m)\,H})/\Tr(\e^{-\beta(E_m)\,H})| \in \mathcal{O}(1/N) \,.
  \end{equation}
\end{hypothesis}
Such formulations of ETH appeared for example in \cite{ETH1,ETH24,DePalma2015,Kim2014}, along with a rigorous proof of a statement that is closely related but weaker than Hypothesis~\ref{hyp:originaleth} for translation invariant Hamiltonians with finite range interactions.
Whether a $1/N$ scaling should be required or whether one would be content with a weaker decay is debatable.

\begin{hypothesis}[Smoothness ETH] \label{hyp:smoothnesseth}
  For any physically reasonable observable $A$ there exists a function $a:\R \to \R$ that is Lipschitz continuous with a Lipschitz constant $L \in \mathcal{O}(1/N)$ such that the expectation values $A_{m} \coloneqq \bra{E_m}A\ket{E_m}$ of $A$ with respect to the energy eigenstates $\ket{E_m}$ in the bulk of the spectrum of a Hamiltonian of system with $N$ particles satisfy
  \begin{equation}
     - \ln |A_{m} - a(E_m)| \in \mathcal{O}(N) \,.
  \end{equation}
\end{hypothesis}
In words: The expectation values of physically reasonable observables in energy eigenstates approximately vary slowly as a function of energy instead of widely jumping over a broad range of values even in small energy intervals.
The function $a(E)$ is often related to the average of $A$ over a small microcanonical energy window around $E$.
Similar statements of the ETH have been used for example in \cite{smoothnessETH1,Bartsch2017,Beugeling2014a,Beugeling2015,Khodja2014,Pure2,Steinigeweg2014,Torres-Herrera2014}. 

Several other versions of the ETH and variations of the statements above can be found in the literature and there is a further level of diversification which needs to be mentioned:
All the statements above are intended to hold for all energy eigenstates in the bulk of the spectrum.
It is also possible to require them to hold only for all but a small fraction of these eigenstates, which somehow goes to zero in the thermodynamical limit.
Such statements have been dubbed \emph{Weak ETH} \cite{WeakETH}.
Another related concept is the \emph{eigenstate randomization Hypothesis} \cite{ERH1}, which states that the diagonal elements of physical observables should behave as random variables.
Together with an assumption on the smoothness of the energy distribution, this allows to derive a bound on the difference between the infinite-time and a suitable microcanonical average.

The main difference among the formulations of the ETH listed above is that the first one is also a statement about the off-diagonal matrix elements $A_{mn}$ while the other two pertain only to diagonal matrix elements $A_{mm}$.
We believe it is important to highlight this aspect because the off-diagonal matrix elements contribute in a non-trivial way to the out-of-equilibrium dynamics of the observable \cite{tsc1,tsc2,tsc3,tsc4,tsc5,tsc6,mondainirigol}.
This is the reason why we (as others do \cite{ETH5}) consider the Original ETH as more fundamental.
Hereafter, when we refer to ETH we will always refer to the technical statement of Original ETH or ETH \ref{hyp:originaleth}.

\paragraph{Physical observables.}
Another issue left open by the above definitions of the ETH is the identification of physical observables for which ETH is supposed to hold.
In this work we show that highly degenerate observables are good candidates.
Those are natural in at least three scenarios:
First, local observables only have a small number of distinct eigenvalues, as they act non-trivially only on a low dimensional space, and each such level is exponentially degenerate in the size of the system on which they do not act.
Second, averages of local observables, like for example the total magnetization, are, for combinatorial reasons, highly degenerate around the center of their spectrum.
Third, \emph{macro observables} as introduced by von~Neumann \cite{Neumann1929,Tumulka2010} and studied in \cite{Typ3,Typ4,Typ5} that are degenerate through the notion of macroscopicity.
Here the idea is that on macroscopically large systems one can only ever measure a rather small number of observables and these observables can take only a number of values that is much smaller than the enormous dimension of the Hilbert space and they commute either exactly or are very close to commuting observables. An example are the classical position and momentum of a macroscopic system.
While they are of course ultimately a coarse grained version of the sum of the microscopic positions and momenta of the all the constituents they can both be measured without disturbing the other in any noticeable way.
Such classical observables hence partition, in a natural way, the Hilbert space of a quantum system in a direct sum of subspaces, each corresponding to a vector of assignments of outcomes for all the macro observables.
Even by measuring all the available macro observables one can only identify which subspace a quantum system is in, but never learn its precise quantum state.
To get the impression that a system equilibrates or thermalizes it is hence sufficient that the overlap of the true quantum state with each of the subspaces from the partition is roughly constant in time and the average agrees with the suitable thermodynamical prediction. 
One would thus expect ETH to hold for such observables.
As in any realistic situation, the number of observables times the maximum number of outcomes per observable (and hence the number of different subspaces) is vastly smaller than the dimension of the Hilbert space, one is again dealing with highly degenerate observables.

\paragraph{Hamiltonian Unbiased Observables.}
Before we proceed with the main result of the paper, it is important to summarize the results derived in \cite{MaxShan}. Suppose $A\coloneqq\sum_i a_i A_i$ is an observable with eigenvalues $a_i$ and respective projectors $A_i$. We say that $A$ is a \textit{thermal observable} with respect to the state $\rho$ if its measurement statistics $p(a_i) \coloneqq \Tr\left( \rho\, A_i\right)$ maximizes the Shannon entropy $S_A \coloneqq - \sum_i p(a_i) \log p(a_i)$ under two constraints: normalization of the state $\Tr(\rho) = 1$ and fixed average energy $\Tr\left(\rho\,H \right)$.

In \cite{MaxShan} it was proven that this is a generalization of the standard notion of thermal equilibrium, in the following sense: What we usually mean by thermal equilibrium is that the state of the system $\rho$ is  close to the Gibbs state $\rho_G$, in the sense given by some distance defined on the convex set of density matrices. A well-known way to characterize $\rho_G$ is via the constrained maximization of von Neumann entropy $S_{\mathrm{vN}}\coloneqq-\Tr( \rho\, \log \rho)$.
Now, for any state $\rho$, the minimum Shannon entropy $S_A$ (among all the observables $A$) is the von Neumann entropy
\begin{align}
\min_{A} S_A = S_{\mathrm{vN}}\,.
\end{align}
Therefore, the Gibbs ensembles is the state that maximizes the lowest among all the Shannon entropies $S_A$.
Hence the maximization of the Shannon entropy $S_A$ is an observable dependent generalization of the ordinary notion of thermal equilibrium. 

One can use the Lagrange multiplier technique to solve the constrained optimization problem and two equilibrium equations emerge. They implicitly define the equilibrium distribution $p_{\mathrm{eq}}(a_i)$ as their solution. Using such equations to investigate the emergence of thermal observables in a closed quantum system, it can be proven that for any given Hamiltonian there is a huge amount of observables that satisfy ETH: the Hamiltonian Unbiased Observables (HUO).

The name originates from the following notion:
A set of normalized vectors $\{\Ket{u_j}\}_j$ is mutually unbiased with respect to another set of vectors $\{\ket{v_k}\}_k$ if the inner product between any pair satisfies |$\braket{u_i}{v_k}| = 1/\sqrt{D}$, where $D$ is the dimension of the Hilbert space. A basis is called Hamiltonian Unbiased Basis (HUB) if it is unbiased with respect to the Hamiltonian basis.
Accordingly, a HUO is an observable which is diagonal in a HUB.
The concept of mutually unbiased basis (MUBs) has been studied in depth in quantum information theory \cite{MUB1,MUB2,MUB3,MUB4,MUB5,MUB6}.
For our purposes, the most important result is the following:
Given a Hilbert space $\Hilb = \otimes_{j=1}^N \Hilb_j$ with $\dim(\Hilb_j) = p$ for some prime number $p$ and some fixed orthonormal basis in $\Hilb$ there is a total of $p^N+1$ orthonormal bases, including the fixed basis, that are all pairwise mutually unbiased \cite{MUB2,MUB3}.
Moreover, there is an algorithm to explicitly construct all of them \cite{MUB2,MUB3}.
Applying this result to the Hamiltonian basis we conclude that there are $p^N$ HUBs.

By studying the matrix elements of a HUO, in the Hamiltonian basis, it is not too difficult to see that sufficiently degenerate HUOs should satisfy ETH (under some mild additional conditions that we discuss in the following.
Suppose a HUO $O^{\mathrm{HUO}}$ has spectral decomposition
\begin{align}
  O^{\mathrm{HUO}} &\coloneqq \sum_{j=1}^{n_A} \lambda_j \Pi_j & \Pi_j &\coloneqq \sum_{s=1}^{d_j} \Ket{j,s}\Bra{j,s} \,.
\end{align}
where $\left\{\Ket{j,s} \right\}$ is the HUB whose elements have been labeled with two indices: $j$ runs over the distinct eigenvalues $\lambda_j$ while $s$ runs over the possible $d_j$ degeneracies of each eigenvalue.
It is easy to see that
\begin{align}
O^{\mathrm{HUO}}_{mm} \coloneqq \bra{E_m}O^{\mathrm{HUO}}\ket{E_m} =   \frac{\Tr(O^{\mathrm{HUO}})}{D}\,.
\end{align}
Therefore, the diagonal matrix elements are constant and the average value at equilibrium, i.e., computed from the diagonal ensemble, is microcanonical 
\begin{align} \label{eq:mc}
\Tr \left(O^{\mathrm{HUO}} \rho_{\mathrm{DE}} \right)= \MV{O^{\mathrm{HUO}}  }_{\mathrm{mc}} \,.
\end{align}
where $\MV{\cdot}_{\mathrm{mc}}$ is the expectation value computed on the microcanonical state $\frac{1}{D}\mathbb{I}$.
Because of the MUB condition we have $\braket{E_m}{j,s} = e^{i\theta^m_{js}} / \sqrt{D}$, which means that the off-diagonal matrix elements are given by
\begin{align}
&\mathcal{O}^{\mathrm{HUO}}_{mn} =  \frac{1}{D}\sum_{j=1}^{n_A} \lambda_j  \sum_{s=1}^{d_j} e^{i\gamma_{js}^{mn}} && \gamma_{js}^{mn} \coloneqq (\theta_{js}^m - \theta_{js}^n)\,.
\end{align}
In \cite{MaxShan} a numerical study on the phases $\gamma_{js}^{mn}$ was performed.
It was argued that the $\gamma_{js}^{mn}$, when constructed with the standard algorithm to build MUBs, have certain features of pseudo-random variables with uniform distribution in $\left[ -\pi , \pi \right]$.
Whenever each eigenvalue has a large degeneracy, i.e., $d_j \gg n_A \geq 2$, we can apply the central limit theorem to argue that
\begin{align}
\mathcal{O}^{\mathrm{HUO}}_{mn} &\approx \sum_{j=1}^{n_A}  X_{mn}^{(j)} & X_{mn}^{(j)} &\sim \mathcal{N} \left[0,\left(\frac{\lambda_j \sqrt{d_j}}{D}\right)^2\right] \label{eqHUO}\,,
\end{align}
where $X^{(j)}_{mn} \sim \mathcal{N}[\mu,\sigma^2]$ means that $X^{(j)}_{mn}$ is a complex random variable, normally distributed, with mean $\mu$ and variance $\sigma^2$.
Under the additional assumption that the $X^{(j)}_{mn}$ are independent, one finds that, because Eq.~\eqref{eqHUO} is a finite sum of normally distributed random variables, we have $\mathcal{O}_{mn}^{\mathrm{HUO}} \sim \mathcal{N} \left[ 0, \sigma^2_{n_A}\right]$ with variance
\begin{equation}
  \sigma_{n_A}^{2} \equiv \sum_{j=1}^{n_A} \left(\frac{\lambda_j \sqrt{d_j}}{D}\right)^2 = \frac{1}{D} \MV{\left(\mathcal{O}^{\mathrm{HUO}}\right)^2}_{\mathrm{mc}}\,,
\end{equation}
Eventually we get:
\begin{align}
\mathcal{O}^{\mathrm{HUO}}_{mn} \approx \sqrt{\frac{1}{D} \MV{\left(\mathcal{O}^{\mathrm{HUO}}\right)^2}_{\mathrm{mc}} }  \,\,\,\mathcal{R}_{mn} \,. \label{eq:HUOandETH}
\end{align}
For a binary observable, i.e, such with eigenvalues $\pm 1$, this means that for large $d_j$ 
\begin{align}
\mathcal{O}^{\mathrm{HUO}}_{mn} &\approx \frac{1}{\sqrt{D}} \mathcal{R}_{mn} & \mathcal{R}_{mn} &\sim \mathcal{N}[0,1] \,,
\end{align}
which means that $\mathcal{O}^{\mathrm{HUO}}$ satisfies Hypothesis~\ref{hyp:originaleth}.

Before we proceed, we would like to expand on the mechanism behind the emergence of ETH for a highly-degenerate HUO.
Eq.~\eqref{eqHUO} will hold whenever we can apply the central limit theorem within each subspace at fixed eigenvalue.
As was argued in \cite{MaxShan}, for a fixed pair of indices $(m,n)$, the phases $\gamma_{js}^{mn}$ behave as if they were pseudo-random variables and their number is exponentially large in the system size.
The labels $(j,s)$ provide a partition of these $D$ phases into $n_A$ groups, each made of $d_j$ elements.
In the overwhelming majority of cases each group of $d_j$ phases will exhibit the same statistical behavior as the whole set.
In this case, Eq.~\eqref{eqHUO} will behave as a sum of independent random variables and it will give the exponential decay of the off-diagonal matrix elements.
It may happen that the index $j$, labeling different eigenvalues, samples the phases in a biased way and prevents some of the off-diagonal matrix elements from being exponentially small.
This, even though it seems unlikely, is possible and it would induce a coherent dynamics on the observable which can prevent its thermalization.
This can happen for example in integrable quantum system for observables which are close to being conserved quantity.

The point can also be seen from the perspective of random matrix theory. Given the Hamiltonian eigenbasis, if we perform several random unitary transformations and study the distribution of the outcome basis, it can be shown that, in the overwhelming majority of cases we will end up with a basis that is almost HUB\cite{MUB5,MUB6}, up to corrections which are exponentially small in the system size.
Hence for large system sizes, if we pick a basis at random, most likely it will be almost a HUB \cite{MUB5,MUB6}.

We now present the main result of the paper: a theorem that can be used to study under which conditions highly degenerate observables are HUO.
\begin{theorem} \label{thm:mainresult}
  Let $\{\ket{\psi_m}\}_{m=1}^M \subset \Hilb$ be a set of orthonormal vectors in a Hilbert space $\Hilb$ of dimension $D$. 
  Let $A = \sum_{j=1}^{n_A} a_j \Pi_j$ be an operator on $\Hilb$ with $n_A \leq D$ distinct eigenvalues $a_j$ and corresponding eigen-projectors $\Pi_j$.
  Decompose $\Hilb = \bigoplus_{j=1}^{n_A} \Hilb_j$ into a direct sum such that each $\Hilb_j$ is the image of the corresponding $\Pi_j$ with dimension $D_j$.
  For each $j$ for which $D_j(D_j-1) \geq  M+1$ there exists an orthonormal basis $\{\ket{j,k}\}_{k=1}^{D_j} \subset \Hilb_j$ such that for all $k,m$
  \begin{equation} \label{eq:THEO}
    | \braket{\psi_m}{j,k} |^2 = \bra{\psi_m} \Pi_j \ket{\psi_m} / D_j \,.  
  \end{equation}
\end{theorem}
A detailed proof is provided in the Supplemental Material \footnote{See Supplemental Material \ref{sm:proof} at [URL will be inserted by publisher] for details of the proof of Theorem~\ref{thm:mainresult}.}.
If the condition $D_j(D_j-1) \geq M+1$ is fulfilled for all $j$, then the set of all $\{\ket{j,k}\}_{j,k}$ obviously is an orthonormal basis for all of $\Hilb$ and $A$ is diagonal in that basis.
So, as long as the degeneracies $D_j$ of $A$ are all high enough with respect to $M$,  $A$ has an eigenbasis whose overlaps with the states $\ket{\psi_m}$ are given exactly by the right hand side of \eqref{eq:THEO}. 

A particularly relevant case is when $A$ is a local observable acting non-trivially only on some small subsystem $S$ of dimension $D_S$ of a larger $N$-partite spin system of dimension $D = d^N$, i.e., $A \coloneqq \sum_{j=1}^{D_S} a_j \ketbra{a_j}{a_j} \otimes \1_{\overline{S}}$ and $\{\ket{\psi_m}\}_{m=1}^M$ is taken to be an eigenbasis $\{\ket{E_m}\}_{m=1}^D$ of the Hamiltonian $H$ of the full system.
We summarize some non-essential further details in the Supplemental Material \cite{Note2}.
In this case the degeneracies are all at least $D_j \geq D/D_S = d^{N-|S|}$, so that the above results guarantees that for all observables on up to $|S| < N/2$ sites there exists a tensor product basis $\{\ket{a_j, k}\}_{j,k}$ for $\Hilb$ which diagonalizes $A$ and with the property that
\begin{equation}
  | \braket{E_m}{a_j,k} |^2 = \frac{1}{d^{N-|S|}} \, \bra{a_j} \Tr_{\overline{S}} \ketbra{E_m}{E_m} \, \ket{a_j} \,.
\end{equation}
For subsystems with support on a small part of the whole system $|S| \ll N-|S|$, it is well known that the reduced states of highly entangled states are (almost) maximally mixed \cite{Typ1}, i.e. proportional to the identity. Moreover, based on the data available in the literature \cite{EE,CFT1,CFT2,CFT3,CFT4,EigEnt1,EigEnt2,EigEnt3,EigEnt3,EigEnt4}, there is agreement on the fact that, away from integrability, the energy eigenstates in the bulk of the spectrum have a large amount of entanglement. Thus, if the eigenstates $\ket{E_m}$ are all highly entangled $ \Tr_{\overline{S}} \ketbra{E_m}{E_m} \approx \1_{S} / d^{|S|}$ and we have
\begin{equation}
  | \braket{E_m}{a_j,k} |^2 \approx 1/d^N \,.
\end{equation}

This way of arguing shows how entanglement in the energy basis can lead the emergence of the ETH in a local observable.
While this result was expected for the diagonal part of ETH, we would like to stress that it is a non-trivial statement about the off-diagonal matrix elements. Since the magnitude of the off-diagonal matrix elements controls the magnitude of fluctuations around the equilibrium values, their suppression in increasing system size is of paramount importance for the emergence of thermal equilibrium.
If one assumes high-entanglement in the energy eigenstates, it is trivial to see that  $A_{mm} \approx \Tr{A}/D$.
Moreover, thanks to the HUO construction and Theorem~\ref{thm:mainresult} we can also make non-trivial statements (Eq.~\eqref{eqHUO} and Eq.~\eqref{eq:HUOandETH}) about the off-diagonal matrix elements.

The physical picture that emerges is the following: Entanglement in the energy eigenstates is the feature which makes a local observable satisfy the statement of the ETH. If the energy eigenstates are highly entangled in a certain energy window $I_0 = \left[E_{a} , E_{b} \right]$, as it is expected to happen in a non-integrable model, the ETH will be true for local observables, in the same energy window. 

We now turn our attention to the study of extensive observables and assume that we are interested in a certain energy window $\left[ E_a,E_b\right] $ which contains $M \leq D$ energy eigenstates.
The details of the computations can be found in the Supplemental Material \footnote{\label{footnote:sm:examples}See Supplemental Material \ref{sm:examples} at [URL will be inserted by publisher] for detailed calculations concerning local, extensive, and macroscopic observables.}.
The paradigmatic case that we study is the global magnetization $M_z \coloneqq \sum_{i=1}^{N} \sigma_i^z$.
Writing its spectral decomposition we have $M_z = \sum_{j=-N}^N j \Pi_j$, where the degeneracy $\Tr \Pi_j = D_j$ of each eigenvalue $j$ can be easily computed to be $D_j = {N \choose \frac{N-j}{2}}$.
Again, we call $\mathcal{H}_j \subset \mathcal{H}$ the image of the projector $\Pi_j$.
The inequality  $D_j(D_j -1) \geq M$ selects a subset $j\in [-j_{*}(M),j_{*}(M)]$ of spaces $\mathcal{H}_j$ for which the conditions of our theorem are satisfied.
Small $M$ will guarantee that the hypothesis of the theorem are satisfied in a larger set of subspaces $\mathcal{H}_j$.
If we are interested in the whole energy spectrum $M=D$, a rough estimation, supported by numerical calculations, shows that $j_{*}(D)$ scales linearly with system size: $j_{*}(D) \simeq 0.78 N$.
The physical intuition that we obtain is the following:
Subspaces with ``macroscopic magnetization'', i.e. around the edges of the spectrum of $M_z$, have very small degeneracy and the theorem does not yield anything meaningful for them.
However, in the bulk of the spectrum there is a large window $j\in [-j_{*}(D),j_{*}(D)]$ where the respective subspaces $\mathcal{H}_j$ meet the conditions for the applicability of the theorem.
Therefore $ \forall j \in \mathbb{Z} \cap [-j_{*}(D),j_{*}(D)] $ we have 
\begin{align}
 |\braket{E_m}{j,s}|^2 = \frac{\bra{E_m} \Pi_j \ket{E_m}}{D_j} \,.
\end{align}
If, for some physical reasons, one is not interested in the whole set energy spectrum but only in a small subset, the window $[-j_{*}(M),j_{*}(M)]$ will increase accordingly. Thanks to our theorem we can extract a physical criterion under which the global magnetization will satisfy ETH. Assuming that we can use Stirling's approximation, the  $M_z$ is a HUO iff 
\begin{equation}
  \bra{E_m} \Pi_j \ket{E_m} \approx 2^{-N H_2(p(j)|\!| p_{\mathrm{mix}})} \label{eq:ldt}\,,
\end{equation}
where $p(j) \coloneqq \left( \frac{1}{2} + \frac{j}{2N}, \frac{1}{2} - \frac{j}{2N}\right)$, $p_{\mathrm{mix}}\coloneqq p(0)$ and we used the binary relative entropy $H_2(p|\!|q) \coloneqq \sum_{k=1,2} p_k \log \frac{p_k}{q_k}$.
This relation has a natural interpretation in terms of large-deviation theory. Indeed, such a relation is a statement about the statistics induced by the energy eigenstates on the observable $M_z$. If such statistics satisfy large-deviation theory, as in Eq.~\eqref{eq:ldt}, the observable will satisfy ETH. A complete understanding of how this concretely happens goes beyond the purpose of the present work and it is left for future investigation.

We note that the hypothesis of the theorem do not hold for the whole spectrum of $M_z$. Moreover, the proven connection between HUOs and ETH relies on the applicability of the central limit theorem in the degeneracy space $\mathcal{H}_j$. Hence the picture that emerges is the following.
For extensive observables, ETH will hold if the statistics induced by the energy eigenstates satisfies a large deviation theory. If this is true, we do not expect it to hold through the whole spectrum but only in the subsectors with sufficiently high degeneracy.
Both statements fully agree with the intuition that, in the thermodynamic limit, macroscopically large values of an extensive sum of local observables should be highly unlikely. In a recent work by Biroli \emph{et al.} \cite{ETH9}, it was argued that in a chain of interacting harmonic oscillators, the measurement statistics of the average of the nearest-neighbor interactions, given by the diagonal ensemble, satisfies a large-deviation statistics. This allows for the presence of rare, non-thermal, eigenstates which can account for the absence of thermalization in some integrable systems. Our results goes along with such intuition. Indeed, if it is possible to show that a large-deviation bound emerges at the level of each energy eigenstate, for all of them, this would amount to a proof of ETH, as discussed before.

We now come to the last application of our theorem: the macro-observables originally proposed by von Neumann.
As for the two previous applications, more details can be found in the Supplemental Material \cite{Note2}.
As explained before, macro-observables induce a partition of the Hilbert space into subspaces in which such classical-like observables have all well defined eigenvalues. In this sense a \emph{macrostate} is an assignment of the eigenvalues of all these observables and the index $j$ runs over different macrostates.
By construction, each macrostate $j=1,\ldots, n$ corresponds to a subspace $\mathcal{H}_j$ of the whole Hilbert space which is highly degenerate and to which we can apply our theorem.
According to the result by von Neumann \cite{Neumann1929} and Goldstein et al. \cite{Typ3} it can be proven that the following relation holds for a given partition, \emph{for most Hamiltonians, in the sense of the Haar measure}: $\bra{E_m} P_j \ket{E_m} = \frac{D_j}{D}$.
The $P_j$'s are the projectors onto the subspaces $\mathcal{H}_j$.
Our theorem tells us that there exists a basis $\left\{ \ket{j,s}\right\}$ which diagonalises all the macro-observables such that $\bra{E_m} P_j \ket{E_m} = D_j |\braket{E_m}{j,s}|^2 $.
Using it in synergy with the previously mentioned result we find: 
\begin{align}
&|\braket{E_m}{j,s}|^2 = \frac{1}{D}\,.
\end{align}
This means that for most Hamiltonians, those macro-observables have a common basis that is a HUB.
Given the huge degeneracy of the spaces $\mathcal{H}_j$ this in turn allows us to formulate the following statement: \emph{for most Hamiltonians, in the sense of Haar, the macro-observables are degenerate HUOs and therefore satisfy ETH \ref{hyp:originaleth}}.

\paragraph{Conclusions.}
The ETH captures the wide spread and numerically very well corroborated intuition that the eigenstates of sufficiently complicated quantum many-body system have thermal properties.
Its importance stems from the fact that together with the results that constitute the framework of pure state quantum statistical mechanics, a proof of the ETH would yield a very general argument for the emergence of not just equilibration, but thermalization towards the prediction of equilibrium statistical mechanics from quantum mechanics alone.
Such a rigorous proof is, however, still missing, despite the progress in recent years that have significantly improved our understanding of the ETH by means of proofs of related statements and counterexamples.
Here we contribute to this program by bisecting the problem of proving ETH in two sub-problems related to the relative phases and the the overlaps between the eigenstates of the Hamiltonian and an observable.
We argue that the ETH can fail because of the former only through conspiratorial correlations in the phases.
Our main result concerns the second half of the problem.
Here we prove a rigorous result that shows when highly degenerate observables satisfy this part of the ETH and become Hamiltonian unbiased observables.
We illustrate our results with three types of physical observables, local, extensive, and macroscopic observables and collect and compare different versions of the ETH.
Our approach allows us in particular to make statements about the off-diagonal elements that are prominent in the original version of the ETH. 

\paragraph{Acknowledgements.}
We would like to thank Charles Xu for interesting discussions.
MH would like to acknowledge funding from the Austrian Science Fund (FWF) through the START project Y879-N27.
FA would like to thank the ``Angelo Della Riccia'' foundation and the St. Catherine's College of Oxford for their constant support to this research.
C.\ G.~acknowledges support by the European Union's Marie Sk\l{}odowska-Curie Individual Fellowships (IF-EF) programme under GA: 700140 as well as financial support from the European Research Council
(CoG QITBOX and AdG OSYRIS), the Axa Chair in Quantum Information Science,
Spanish MINECO (FOQUS FIS2013-46768, QIBEQI FIS2016-80773-P and Severo Ochoa Grant No.~SEV-2015-0522), EU STREP program EQuaM
(FP7/2007-2017, Grant No. 323714), Fundaci\'{o} Privada Cellex, and Generalitat de Catalunya (Grant No.~SGR 874 and 875, and CERCA Programme).
Furthermore we would like to express gratitude towards the COST Action MP1209 ``Thermodynamics in the quantum regime'', whose workshop in Smolenice sparked this collaboration.

\makeatletter
\newcommand{\manuallabel}[2]{\def\@currentlabel{#2}\label{#1}}
\makeatother

\clearpage
\appendix
\onecolumngrid

\manuallabel{sm:proof}{Section~A}
\section{\ref*{sm:proof}: Proof of the main theorem}

In this Appendix we provide the details of the proof of the main result of the paper, Theorem~\ref{thm:mainresult}. In the first subsection we give some background material, concerning the formalism of the generalized Bloch-vector parametrization. Such formalism will be used in the second subsection, where we give the actual proof of Theorem~\ref{thm:mainresult}.

\manuallabel{sm:generalizedbloch}{Subsection~A1}
\subsection{\ref*{sm:generalizedbloch}: Generalised Bloch-vector parametrization}
We start by briefly recalling the formalism of the generalized Bloch-vector parametrization\cite{geom,bloch} of a pure quantum state. The standard Bloch-vector parametrization is a well-known way to describe the space of pure-states of a qubit, by using the isomorphism between its two-dimensional projective Hilbert space and a 2-sphere $\mathbb{S}^2$. Such an isomorphism can be easily generalized to arbitrary dimensions and it is well known that the projective space of a $D-$dimensional complex Hilbert space is isomorphic to $\mathbb{S}^{D^2-2}$. This isomorphism can be made explicit by associating to any normalized rank-$1$ projector $\ketbra{\psi}{\psi}$ a generalized Bloch vector $\vec{b}(\psi) \in \mathbb{S}^{D^2-2} \subset \mathbb{R}^{D^2 - 1}$ that fulfills
\begin{align}\label{eq:blochvectorcorrespondence}
\ketbra{\psi}{\psi} = \frac{\mathbb{I}}{D} + \sqrt{\frac{D-1}{D}} \,\, \vec{b}(\psi) \cdot \vec{\gamma}
\end{align}
where $\vec{\gamma}$ is a vector with elements $\gamma_i \coloneqq \hat{\gamma}_i / \sqrt{2}$ and $\hat{\gamma}_i$ are the $D^2-1$ generators of $SU(D)$, with the following properties:
\begin{align}
\hat{\gamma}_i &= \hat{\gamma}_i^\dagger & \Tr(\hat{\gamma}_i) &= 0 & \Tr(\hat{\gamma}_i\,\hat{\gamma}_j) &= 2\, \delta_{ij}
\end{align}
Even though the term ``Bloch vector'' is normally used to identify the $2$-dimensional case, hereafter we will use it for its $D$-dimensional counterpart. The constant prefactor $\sqrt{\frac{D-1}{D}}$ has been put to make the norm of the Bloch vector independent on the dimension of the Hilbert space and always equal to one. The square of the absolute value of the scalar product between two pure states $\Ket{\psi},\Ket{\psi'} \in \Hilb$ is mapped into the scalar product of the two Bloch vectors $\vec{b},\vec{b}'$, plus a constant term
\begin{align} 
  |\braket{\psi}{\psi'}|^2 = \frac{1}{D} + \frac{D-1}{D} \,  \vec{b} \cdot \vec{b}' .
\end{align}

From this relation we can see that mutual unbiasedness is a very natural condition when written in term of the respective Bloch vectors. For any two sets of pure states $\{\Ket{\psi_j}\}_j$ and $\{\ket{\psi'_k}\}_k$, with respective Bloch vectors $\{\vec{b}_j\}_j$ and $\{\vec{b}'_k\}_k$ we have 
\begin{equation}
  |\braket{\psi_j}{\psi'_k}|^2 = \frac{1}{D} \quad \Longleftrightarrow \quad \vec{b}_j \cdot \vec{b}'_k = 0 \quad .
\end{equation}
In other words the sets $\{\Ket{\psi_j}\}_j$ and $\{\ket{\psi'_k}\}_k$ are mutually unbiased if and only if their respective sets of Bloch vectors are orthogonal. Now we look at how the property 
of being a basis of the Hilbert space is written in terms of the Bloch vectors of the basis elements. Let $\{\ket{\psi_j}\}_{j=1}^{D} \subset \Hilb$ be a basis of a Hilbert space of dimension 
$D$, with associated Bloch vectors $\{\vec{b}_j\}_j$. Using Eq.~\eqref{eq:blochvectorcorrespondence} we find that $\{\ket{\psi_j}\}_{j=1}^{D}$ spans all of $\Hilb$ if and only if
\begin{align}
  \1 = \sum_{j=1}^{D} \ketbra{\psi_j}{\psi_j} = \1 + \sqrt{\frac{D-1}{D}} \sum_{j=1}^{D} \vec{b}_j \cdot \vec\gamma .
\end{align}
Since the elements of $\vec\gamma$ are the linearly independent generators of $SU(D)$, this is equivalent to $\sum_{j=1}^{D} \vec{b}_j = 0$.
At the same time, the vectors $\{\ket{\psi_j}\}_{j=1}^{D}$ are orthonormal if and only if $\forall j,k \in \{1,\dots,D\}$
\begin{align}
  \delta_{jk} = |\Scal{\psi_j}{\psi_k}|^2 = \frac{1}{D} + \frac{D-1}{D} \, \vec{b}_j \cdot \vec{b}_k ,
\end{align}
which is equivalent to
\begin{align}
 \vec{b}_j \cdot \vec{b}_k = \frac{D}{D-1} \delta_{jk} - \frac{1}{D-1} .
\end{align}
In summary we obtain that $\{\ket{\psi_j}\}_{j=1}^{D}$ is a complete orthonormal basis if and only if their Bloch vectors $\{\vec{b}_j\}_j$ satisfy the two following conditions
\begin{subequations}
\begin{align} \label{eq:conditionsforbeingabasis}
  &&\sum_{k=1}^{D} \vec{b}_k &= 0\\
  &\text{and} & \vec{b}_h \, \cdot \, \vec{b}_k &= \frac{D}{D-1} \delta_{hk} - \frac{1}{D-1} =
                                    \begin{cases}
                                      1 & \text{if } $h=k$\\
                                      -\frac{1}{D-1} & \text{if } h \neq k 
                                    \end{cases}
\end{align}\label{eq:BASE} 
\end{subequations}

\manuallabel{sm:proofofmaintheorem}{Subsection~A2}
\subsection{\ref*{sm:proofofmaintheorem}: Proof of Theorem~\ref{thm:mainresult}}\label{sec:proofofmainresult}
In this second Appendix we present a detailed proof of Theorem~\ref{thm:mainresult} from the main text.
In order to do this we first introduce a well known theorem from geometry and the notions necessary to state it.
We then show how the generalized Block vector parametrization together with this theorem and properties of simplices allow to prove Theorem~\ref{thm:mainresult}.

In $\mathbb{R}^n$ an $n$-simplex is the generalization of the $2D$ triangle and the $3D$ tetrahedron to arbitrary dimensions.
A \emph{regular simplex} is a simplex which is also a regular polytope.
For example, the regular $2$-simplex is the equilateral triangle and the regular $3$-simplex is a tetrahedron in which all faces are equilateral triangles.
A $n$-simplex can be constructed by connecting a new vertex to all vertices of an $n-1$-simplex with the same distance as the common edge distance of the existing vertices.
This readily implies that the convex hull of any subset of $n$ out of the $n+1$ vertices of an $n$ simplex is itself a $n-1$-simplex, a so called \emph{facet} of the simplex.    
For $n=2$ they are the sides of the triangle, for $n=3$ they are the two dimensional triangles building the boundary surface of the tetrahedron.
To each facet we can associate a \emph{facet vector} defined as the vector orthogonal to the facet and with Euclidean length equal to the volume of the facet.
The result we need about these objects is the following theorem.

\begin{theorem}[Minkowski(-Weyl) Theorem \cite{poly}] \label{thm:minkowskiweyl}
  For any set of $n+1$ non co-planar vectors $\vec{V}_i \in \R^n$ that span $\R^n$ with the property 
  \begin{align}\label{eq:sumzero}
    \sum_{i=1}^{n+1} \vec{V}_i = 0 
  \end{align}
  there is a closed convex $n$-dim polyhedron whose facet vectors are the $\vec{V}_i$.  The converse is also true, for any closed convex polyhedron the facets vectors sum to zero.
\end{theorem}

If we apply the theorem to an $n$-simplex, whose facets vector are all of equal magnitude it can be easily seen that the (all equal) dihedral angles $\alpha$ between two facet vectors
are such that $\cos \alpha = - \frac{1}{n}$. This fact will be used in the proof of Theorem~\ref{thm:mainresult}. Projecting Eq.~\eqref{eq:sumzero} onto the direction of one vector $\vec{V}_k$ and using the fact that all dihedral angles have the same magnitude $\alpha$ in a simplex we have $\sum_{i=1}^{n+1} \vec{V}_k \cdot \vec{V}_i = 1 + n \cos \alpha = 0$. Which gives $\cos \alpha = - \frac{1}{n}$. We can now proceed with the proof of Theorem~\ref{thm:mainresult}.

\begin{proof}[Proof of Theorem~\ref{thm:mainresult}]
  If $\mathcal{H}$ is a $D$-dimensional Hilbert space, take an arbitrary decomposition $\mathcal{H}=\oplus_{j=1}^n \mathcal{H}_j$ and call $P_j$ the projectors onto $\mathcal{H}_j$. Define $p_{m,j} \coloneqq \bra{\psi_m}\,P_j\,\ket{\psi_m}$.
  For every $\ket{\psi_m}$ let
  \begin{equation}
    \ket{\psi_m^{(j)}} \coloneqq
    \begin{cases}
      P_j\,\ket{\psi_m}/\sqrt{p_{m,j}} & \text{if } p_{m,j} \neq 0\\
      0 & \text{otherwise}
    \end{cases}
  \end{equation}
  be the normalized projection onto the subspace associated with $P_j$ or the zero vector if $\ket{\psi_m}$ is orthogonal to that subspace.
  Now, for any vector $\ket{\varphi} \in \Hilb_j$ we can write $| \braket{\psi_m}{\varphi} |^2 = | \bra{\psi_m}\,P_j\,\ket{\varphi} |^2 = p_{j,k}\, | \braket{\psi_m^{(j)}}{\varphi} |^2$.
  As both $\ket{\psi_m^{(j)}}$ and $\ket{\varphi}$ are contained in $\Hilb_j$, via the construction described in \ref{sm:generalizedbloch}, they have associated generalized Bloch vectors $\vec{b}_m^{(j)}$ and $\vec{b}$ in $\mathbb{S}^{D_j^2-2}$.
  Using Eq.~\eqref{eq:blochvectorcorrespondence} we thus have
  \begin{equation}
    | \braket{\psi_m}{\varphi} |^2 = p_{m,j} \, \frac{1}{D_j} + p_{m,j} \, \frac{D_j-1}{D_j} \,  \vec{b} \cdot \vec{b}_m^{(j)} .
  \end{equation}
  We conclude that $\ket{\varphi} \in \Hilb_j$ has the desired property (Eq.~\eqref{eq:THEO}) of the basis vectors $\ket{j,k}$ if and only if $\vec{b}$ is orthogonal to all the $\vec{b}_m^{(j)}$.
  For any given $j$, in the worst case, all the $M$ vectors $\vec{b}_m^{(j)}$ are linearly independent, leaving a subspace of dimension $D_j^2-2-M$ for picking $\vec{b}$.
  Now, we don't want to pick just one vector $\vec{b}$ from this subspace, but $D_j$ many such vectors, which moreover satisfy the conditions in \eqref{eq:conditionsforbeingabasis} so that their associated state vectors form an orthonormal basis for $\Hilb_j$.
  The Minkowski(-Weyl) Theorem (Theorem~\ref{thm:minkowskiweyl}) tells us that this can be achieved by taking them to be the facet vectors $\vec{V}_i$ of a regular simplex in this subspace, as long as the subspace has sufficiently high dimension.
  More precisely, the first condition from \eqref{eq:conditionsforbeingabasis} is always satisfied for facet vectors $\vec{V}_i$ of general polytopes and the second condition can be achieved by using the facet vectors of a regular simplex, scaled so that they have Euclidean norm equal to one.
  This follows because the cosine of the angle between any two facet vectors of an $n$-simplex is $-1/n$.
  So, as long as the space of vectors orthogonal to all the $\vec{b}_m^{(j)}$ is large enough to accommodate for a $D_j-1$-simplex, $D_j$ suitable Bloch vectors of an orthonormal basis $\{\ket{j,k}\}_{k=1}^{D_j} \subset \Hilb_j$ that is unbiased with respect to all $\ket{\psi_m}$ can be found.
  This is the case as long as $D_j^2-2-M \geq D_j-1$.
\end{proof}

\manuallabel{sm:examples}{Section~B}
\section{\ref*{sm:examples}: Examples}
In this second Appendix, we give more details about how to apply Theorem~\ref{thm:mainresult} to the three examples given in the manuscript and how to derive the results. We use a one-dimensional spin-1/2 chain as an exemplary case to showcase our result. Moreover, we will always be interested in using the Hamiltonian eigenvectors as a set of vectors for our theorem. This means $M=D$ and $\left\{ \ket{\psi_j} \right\}_{j=1}^D = \left\{ \ket{E_m}\right\}_{m=1}^D$. However, if for some reason one is interested in a limited portion of the energy spectrum, the results can be strengthened by limiting the set of eigenvectors to $M < D$.

\subsection{Example 1: Local observables}

As first application of our Theorem, we study the emergence of ETH in a local observable which has support on less than half of the whole chain. The total number of spins is $N$ and the Hilbert space is split into tensor products of $k$ and $N-k$ spins: $\mathcal{H} = \mathcal{H}_k \otimes \mathcal{H}_{N-k}$. Local observables $A_{\mathrm{loc}}= A_k \otimes \mathbb{I}_{N/k} = \sum_{j=1}^{2^k} P_j a_j$ have support on $k \leq N-k$ sites. In this case all eigenvalues have degenerate subspaces with the same dimension: $\mathrm{dim} \mathcal{H}_j = \Tr P_j = D_j = 2^{N-k}$. The condition that ensures the validity of the hypothesis of Theorem~\ref{thm:mainresult} is $2^{N-k}(2^{N-k}-1) \geq 2^N + 1$. Applying the $\log $  to both sides and with some algebraic manipulations we obtain
\begin{equation}
2(N-k)\log 2 - N \log 2 \geq \log \left( \frac{1-\frac{1}{2^{N-k}}}{1+\frac{1}{2^N}}\right)
\end{equation}
The right-hand side is always negative. So we request the following (slightly stronger) condition
\begin{equation}
2(N-k)\log 2 - N \log 2 \geq 0 \geq \log \left( \frac{1-\frac{1}{2^{N-k}}}{1+\frac{1}{2^N}}\right)
\end{equation}
The condition arising from the first inequality gives $k \leq \frac{N}{2}$. Therefore, local observables with support on less than half of the chain satisfy the assumptions of our theorem. For them we obtain that there is a basis $\ket{a_j,k}$ that diagonalizes the observable, such that

  \begin{equation} 
    | \braket{E_m}{a_j,k} |^2 = \bra{E_m} P_j \ket{E_m} / 2^{N-k}
  \end{equation}
since $P_j = A_j \otimes \mathbb{I}_{N/k}$ we have $\bra{E_m} P_j \ket{E_m} = \Tr_k \left(A_j \rho_k(E_m)\right)$ where $\rho_k(E_m) = \Tr_{N/k}\ket{E_m}\bra{E_m}$. For small subsystems $k \ll N-k$, if the Hamiltonian eigenstates are highly entangled, which is expected to be true for a non-integrable system in the bulk of the spectrum, the von Neumann entropy of the reduced state is close to the maximum value $k \log 2 - S_{\mathrm{vN}}(\rho_k(E_m)) \leq \epsilon_{k}(E_m)$ with $\epsilon_k(E_m) \geq 0$. Using Pinsker's inequality and the fact that the relative entropy with respect to the maximally mixed state is just the difference between the two entropies we have 
\begin{equation}
|| \rho_k(E_m)- \frac{\mathbb{I}}{2^k}||^2 \leq \frac{1}{2} (k \log 2 - S_{\mathrm{vN}}(\rho_k(E_m))) \leq \frac{\epsilon_k(E_m)}{2} \, .
\end{equation}
Whenever $\epsilon_k(E_m) \ll 1$, which is expected to be true in the bulk of the spectrum, we have
\begin{equation}
| \braket{E_m}{a_j,k} |^2 = \frac{\bra{E_m} P_j \ket{E_m} }{2^{N-k}} = \frac{\Tr_k A_j \rho_k(E_m)}{2^{N-k}} = \frac{\bra{a_j} \rho_k(E_m)\ket{a_j}}{2^{N-k}} \simeq \frac{1}{2^N} \, .
\end{equation} 
We can therefore conclude that entanglement in the energy eigenstate is the feature that makes local observables be HUOs. Provided certain mild assumptions, which have been discussed in the paper, are satisfied, this guarantees that they satisfy ETH. We conclude that, if the energy eigenstates are highly entangled in a certain energy window $I_0 = \left[E_{a} , E_{b} \right]$, as it is expected to happen in a non-integrable model, ETH will hold for all local observables, in the same energy window.

\subsection*{Example 2: Extensive observable - Global magnetization }
In this second example we study the consequences of our theorem for an observable which is the extensive sum of local observables: the global magnetization $M_z=\sum_{n=1}^N \sigma_n^z$. Its spectral decomposition is $M_z=\sum_{j=-N}^{N} j \, P_j$ so the Hilbert space is decomposed as the direct sum of the $\mathcal{H}_j$, which are the images of the $P_j$: $\mathcal{H}=\bigoplus_{j=-N}^{N} \mathcal{H}_j$. Their dimension $D_j = \Tr \, P_j$ can be computed using combinatorial arguments: $D_j = C^N_{j}\equiv {N \choose \frac{N-j}{2}}$.  At fixed size $N$, $D_j \in [1,{N \choose N/2}]$. The inequality $D_j \geq 1+\frac{2^N+1}{D_j}$ selects a subset $j \in [-j_{*} ,j_{*} ]$ of subspaces $\mathcal{H}_j$ for which the theorem will hold. Note that the interval $[-j_{*},j_{*}]$ is symmetric with respect to zero because $D_{j}=D_{-j}$. In order to find how $j_{*}$ scales with the system size, we numerically compute how many subspaces $\mathcal{H}_j$ meet the condition $D_j \geq 1+\frac{2^N+1}{D_j}$. We call this number $q(N)$. Since the eigenvalues are given by the relative number $j \in \mathbb{Z} \cap [-N,N]$ and they are equally spaced, we have $q(N)=2 j_{*}+1$. Which means $j_{*}=\frac{q(N)-1}{2}$. In Fig.\ref{fig:Macro} we can see that it scales linearly with the system size: $q(N) \sim 1.56 N$. This gives $j_{*}(N) \sim 0.78 N$.

\begin{figure}[h]
\includegraphics[width=3.6 in]{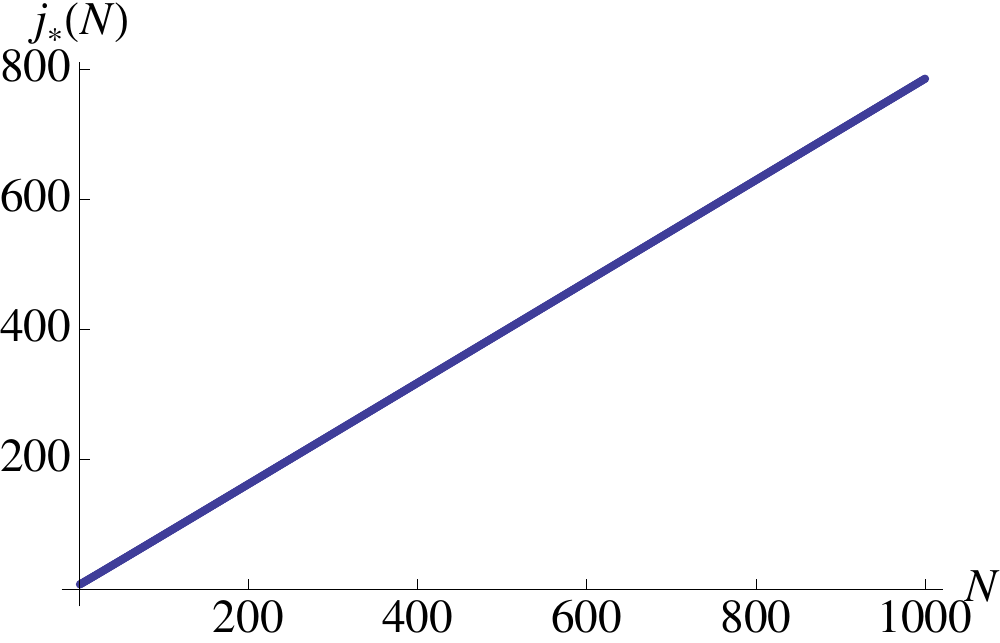}
\caption{Scaling of the number of subspaces $\mathcal{H}_j$ which meet the condition $D_j \geq 1+ \frac{2^N+1}{D_j}$.}\label{fig:Macro}
\end{figure}
The picture that we obtain is the following. States with ``macroscopic magnetization'', i.e. around the edges of the spectrum of $M_z$, have very small degeneracy and the theorem is not going to hold for them. In the bulk of the spectrum, however, there is a large window $j\in [-j_{*}(N),j_{*}(N)]$ where the respective subspaces $\mathcal{H}_j$ meet the conditions for the validity of the theorem. In summary, if we apply the theorem to the global magnetization we obtain:

\begin{equation}
\forall j \in \mathbb{Z} \cap [-j_{*}(N),j_{*}(N)]\,, \qquad  |\braket{j,s}{E_m}|^2 = \frac{\bra{E_m} P_j \ket{E_m}}{D_j} \,.
\end{equation}
We know that the relation we are interested in is the Hamiltonian Unbiasedness, which would be $\frac{\bra{E_m} P_j \ket{E_m}}{D_j} \simeq \frac{1}{2^N}$. For this reason we study the relation
\begin{equation}
\frac{\bra{E_m} P_j \ket{E_m}}{D_j} \simeq \frac{1}{2^N} \qquad \Rightarrow \qquad  \bra{E_m} P_j \ket{E_m} \simeq \frac{D_j}{2^N}\,,
\end{equation}
which in turn means to study how $\frac{D_j}{2^N}$ behaves. For this goal, in the large $N$ regime we can use Stirling's approximation. As it is known, there is not a unique way of using it. Rather, there are different ways, depending on the number of sub-leading terms that one is willing to use. Here we focus on the leading term. Note that Stirling's approximation can be used throughout the whole window $[-j_{*}(N),j_{*}(N)]$, as long as $N \gg 10$. This is true because $j_{*}(N) \sim 0.78 N$ so $|j| \in [0,0.78 N]$ and $\frac{N-j}{2}\sim 0.1*N$. Therefore, as long as $0.1 N \gg 1$, we can use Stirling's formula for all the factorials involved in $D_j$. It can be shown that if $n\geq k \gg 1$, at the leading order we have ${n \choose k} \sim 2^{n H_2(\frac{k}{n})}$ where $H_2(x)\equiv -x\log_2 x - (1-x) \log_2 (1-x)$ is the binary entropy. Using this we get 
\begin{equation}
D_j \approx 2^{N\,H_2\left(\frac{N-j}{2N}\right)} = 2^{N\,H_2\left(\frac{1}{2} - \frac{j}{2N}\right)} \,\,\, ,
\end{equation}
which in turn gives 
\begin{equation}
\frac{D_j}{2^N} \approx 2^{-N [1-H_2\left(\frac{1}{2} - \frac{j}{2N}\right)]} \,\,\, .
\end{equation}
We have the size of the system $N$ which multiplies a function which is a binary
relative entropy. If we call $p_{\mathrm{mix}}\coloneqq \left( \frac{1}{2}, \frac{1}{2} \right)$ and $p(j)=\left(\frac{1}{2} - \frac{j}{2N} ,\frac{1}{2} + \frac{j}{2N} \right)$ we have 

\begin{equation}\label{eq:ldb}
1-H_2[\left(\frac{1}{2} - \frac{j}{2N}\right)] = H_2\left[ p(j)||p_{\mathrm{mix}} \right] \quad \Rightarrow \quad \frac{D_j}{2^N} \approx 2^{-N H_2\left[ p(j)||p_{\mathrm{mix}} \right]}\,.
\end{equation}
Eq.~\eqref{eq:ldb} has a very interesting form.
It is telling us that the statistics of the eigenvalues $a_j$, induced by the eigenstates $\ket{E_m}$, satisfies a large deviation bound. The rate function is given by the binary Kullback-Leibler divergence $H_2[p(j)|\!| p_{\mathrm{mix}}]$. Now we can formulate a clear statement. Choose a subspace $\mathcal{H}_j$ with $|j| < j_{*}$ where the hypothesis of our theorem hold. \emph{If there is a $k \in \mathbb{N}, \,\,\, k < j_{*}$ such that for all $j \in [-k,k]$ we have $\bra{E_m}P_j \ket{E_m} \approx 2^{-N H_2\left[ p(j)||p_{\mathrm{mix}} \right]}$, the global magnetization $M_z$ will be an HUO and satisfy the ETH in the subspaces $\bigoplus_{|j|<k} \mathcal{H}_j$.} Concretely, this will happen if the measurement statistics generated by the energy eigenstates $\ket{E_m}$ on the eigenvalues $a_j$ satisfies a large deviation bound.

To build our intuition on what this means we evaluate $H_2\left(\frac{N-j}{2N}\right)$ in two regimes allowed by our Theorem: $\frac{|j|}{N} \ll 1$ and $\frac{|j| - j_{*}}{N} \ll 1$. In the first case, calling $x=\frac{|j|}{N}$ we can Taylor-expand $H_2(\frac{1-x}{2} )$ around $x \ll 1$ to obtain 
\begin{equation}
H_2\left(\frac{1-x}{2}\right) \stackrel{x \ll 1}{\approx} 1 - \frac{x^2}{2} \qquad \Rightarrow \qquad H_2\left(\frac{1-|j|/N}{2}\right) \approx 1 - \frac{j^2}{2N^2} \qquad |j|/N \ll 1\,. \label{eq:expansion}
\end{equation}

In the regime $|j|\approx j_{*}$ we have a better way to estimate $D_j$. Indeed in such regime $D_j \approx D_{j_{*}}$, which satisfies $D_{j_{*}} \approx 1+ \frac{2^N-1}{D_{j_{*}}}$. Solving for $D_{j_{*}}$ and taking the leading order in $N$ we obtain $D_{j_{*}} \approx 2^{N/2}$. Moreover, using the expression in Eq.~\eqref{eq:expansion} we can find how $D_j$ deviates from $D_{j_{*}}$. Indeed expanding $H_2(\frac{1-|j|/N}{2})$ around $j_{*}$ we get
\begin{equation}
H_2\left(\frac{1-|j|/N}{2}\right) \approx H_2\left(\frac{1-j_{*}/N}{2}\right) - \left. \frac{dH_2}{dx} \right\vert_{x=\frac{1-j_{*}/N}{2}} \frac{|j|-j_{*}}{2N} \approx \frac{1}{2} - \frac{3}{2} \left(\frac{|j|-j_{*} }{N}\right)\,.
\end{equation}
In summary, when $N \gg 10$
\begin{align}
& \frac{D_j}{2^N} \approx  \left\{ \begin{array}{ll} 
 2^{- \frac{j^2}{2N}} & \quad \frac{|j|}{N} \ll 1  \\
& \\
 2^{-\frac{N}{2} - \frac{3}{2}(|j|-j_{*})} & \quad \frac{|j| - j_{*}}{N} \ll 1
  \end{array} \right. 
\end{align}
This means that when we approach the thermodynamic limit $N \to \infty$, the eigenvalues with higher magnetization will be exponentially suppressed in the system size. This is indeed what we expect to be true at the macroscopic level.

\subsection*{Example 3: Macroscopic equilibrium - Normal typicality and von Neumann's Quantum H-theorem}

In this last example we investigate the connection of our theorem with the notion of Macro-observables proposed by von Neumann in his work on the Quantum H-theorem \cite{Neumann1929,Tumulka2010}.
This in turn is strictly related with the notion of Normal typicality developed in a series of more recent works by Goldstein \emph{et al.} \cite{Typ3,Typ4,Typ5}.
Again, we start by decomposing our Hilbert space $\mathcal{H}$ as a direct sum of subspaces $\mathcal{H}_j$.
The index $j$ runs over a finite number of values that identify different macroscopic properties of the system.
One could say that it identifies different ``macrostates'', characterized by the expectation value of commuting macroscopic observables.
In the original idea by von Neumann, in a classical system we measure position and momentum, which commute.
His point was that there are some coarse-grained approximation of actual position and momenta which can be ``rounded'' to obtain a set of commuting macro-observables.
Such set of commuting Macro-observables provides a decomposition of the Hilbert space $\mathcal{H} = \bigoplus_{j=1}^n \mathcal{H}_j$ where the index $j$ runs over all the possible different macrostates.
Each one of these spaces $\mathcal{H}_j$ is hugely degenerate and we assume here that we can use our Theorem for all of them.
Using the concentration of measure phenomenon it can be shown \cite{Typ3,Typ4,Typ5} that for most $t$, $\bra{\psi(t)}P_j \ket{\psi(t)} \simeq \frac{D_j}{D}$ for all $j$, for most Hamiltonians in the sense of Haar and for all $\psi(0)$.

Concretely, this happens for all $\psi(0)$ if and only if $\bra{E_m}P_j \ket{E_m} \simeq \frac{D_j}{D}$ for all $j$ and $m$.
Such a relation can be proven to hold in the same sense as before.
For most Hamiltonians in the sense of Haar
\begin{equation}
  \bra{E_m}P_j \ket{E_m} \simeq \frac{D_j}{D} \qquad \forall  j,m\,.
\end{equation}
The unitary for which this ``most'' holds is the one connecting the Hamiltonian eigenbasis to the basis giving the decomposition of the Hilbert space into ``commuting macro-observables''. We can now see that the connection of these ideas with ETH is unraveled by our theorem 1 and by the notion of HUO. Indeed, using our theorem, we can write 
\begin{equation}
\bra{E_m}P_j \ket{E_m} = D_j \left| \braket{j,s}{E_m}\right|^2
\end{equation}
Therefore
\begin{equation}
  \bra{E_m}P_j \ket{E_m} \simeq \frac{D_j}{D} \qquad \Longleftrightarrow \qquad \left| \braket{j,s}{E_m}\right|^2 \simeq \frac{1}{D}
\end{equation}
From this we conclude that \emph{for most Hamiltonians, in the sense of Haar, that the basis $\left\{\ket{j,s}\right\}$ which diagonalizes all the ``commuting macro-observables'' giving the decomposition $\mathcal{H}=\bigoplus_j \mathcal{H}_j$ is an Hamiltonian Unbiased Basis (HUB)}.
Moreover, thanks to the fact that each subspace $\mathcal{H}_j$ is highly degenerate and that the decomposition $\mathcal{H}=\oplus_j \mathcal{H}_j$ is generated by Macro-observables, this proves that all Macro-observables built in this way are HUO. Again, provided certain mild assumptions, which have been discussed in the main text, are satisfied, this guarantees that they satisfy ETH.


\begin{thebibliography}{50}
\providecommand{\natexlab}[1]{#1}
\providecommand{\url}[1]{\texttt{#1}}
\expandafter\ifx\csname urlstyle\endcsname\relax
\providecommand{\doi}[1]{doi: #1}\else
\providecommand{\doi}{doi: \begingroup \urlstyle{rm}\Url}\fi  

\bibitem{Pure3}  V. I. Yukalov, {\it Equilibration and thermalization in finite quantum systems}, \href{https://doi.org/10.1002/lapl.201110002}{Laser Phys. Lett. {\bf 1} 435 (2004)}.

\bibitem[Cazalilla and Rigol(2010)]{Cazalilla2010a}
M.~A. Cazalilla and M.~Rigol,
{\it Focus on dynamics and thermalization in isolated quantum many-body
  systems},
\href{https://doi.org/10.1088/1367-2630/12/5/055006}{{New Journal of Physics}, {\bf 12}, (2010)}.

\bibitem{Pure2} A. Polkovnikov, K. Sengupta, and M. Vengalattore, {\it Colloquium: Nonequilibrium dynamics of closed interacting quantum systems}, \href{https://doi.org/10.1103/RevModPhys.83.863}{Rev. Mod. Phys. {\bf 83}, 863 (2011).}

\bibitem{Pure1} J. Eisert, M. 
Friesdorf, C. Gogolin, {\it Quantum many-body systems out of equilibrium}, \href{https://doi.org/10.1038/nphys3215}{Nature Physics {\bf 11}, 124D130 (2015)}.

\bibitem{Pure1b} C. Gogolin, J. Eisert, {\it Equilibration, thermalisation, and the emergence of statistical mechanics in closed quantum systems}, \href{https://doi.org/10.1088/0034-4885/79/5/056001}{Rep. Prog. Phys. {\bf 79}, 056001 (2016)}.


\bibitem{Neumann1929} J. v. Neumann, {\it Beweis des Ergodensatzes und des H-Theorems in der neuen Mechanik}. \href{https://doi.org/10.1007/BF01339852}{Zeitschrift Fur Physik, {\bf 57}, 30-70. (1929)}.

\bibitem{Typ6} J. Gemmer, A. Otte \& G. Mahler, {\it Quantum approach to the second law of thermodynamics}, \href{https://doi.org/10.1103/PhysRevLett.86.1927}{Phys. Rev. Lett. {\bf 86}, 1927 (2001)}.

\bibitem{Typ1} S. Popescu, A. Short, A. Winter, {\it Entanglement and the foundations of statistical mechanics}. \href{https://doi.org/10.1038/nphys444}{Nature Physics {\bf 2}, 754 - 758 (2006)}.

\bibitem{Typ3} S. Goldstein, J. L. Lebowitz, R. Tumulka, N. Zanghi, {\it Canonical typicality}. \href{https://doi.org/10.1103/PhysRevLett.96.050403}{Phys. Rev. Lett. {\bf 96}, 050403 (2006)}

\bibitem{Typ2} J. Gemmer, M. Michel \& G. Mahler, \href{https://doi.org/10.1007/978-3-540-70510-9}{{\it Quantum Thermodynamics}, Lecture Notes in Physics: Springer Berlin Heidelberg, 2009}.


\bibitem{Tumulka2010} J. v. Neumann, {\it Proof of the ergodic theorem and the H-theorem in quantum mechanics}, \href{https://doi.org/10.1140/epjh/e2010-00008-5}{J. EPJ H, {\bf 35}, 201-237. (2010)}.


\bibitem{Typ4} S. Goldstein, J. L. Lebowitz, R. Tumulka, N. Zanghi, {\it Long-time behavior of macroscopic quantum systems},  \href{https://doi.org/10.1140/epjh/e2010-00007-7}{Eur. Phys. J. H {\bf 35}, 173 (2010).}

\bibitem{Typ5} S. Goldstein, J. L. Lebowitz, C. Mastrodonato, R. Tumulka, N. Zanghi, {\it Normal typicality and von Neumann's quantum ergodic theorem}, \href{https://doi.org/10.1098/rspa.2009.0635}{Proc. R. Soc. A {\bf 466}, 3203 (2010).}
  
\bibitem{Rei1} P. Reimann, {\it Foundation of Statistical Mechanics under Experimentally Realistic Conditions}, \href{https://doi.org/10.1103/PhysRevLett.101.190403}{Phys. Rev. Lett. {\bf 101}, 190403 (2008).}

\bibitem{tsc6} N. Linden, D. Popescu, A. Short and A. Winter, {\it Quantum mechanical evolution towards thermal equilibrium}, \href{https://doi.org/10.1103/PhysRevE.79.061103}{Phys. Rev. E, 79(6), 61103 (2009)}. 

\bibitem{Rei2} P. Reimann, {\it Canonical thermalization}, \href{https://doi.org/10.1088/1367-2630/12/5/055027}{New J. Phys. {\bf 12} 055027 (2010).}

\bibitem{tsc4}  A. J. Short and T. C. Farrelly, {\it Quantum equilibration in finite time}, \href{https://doi.org/10.1088/1367-2630/14/1/013063}{New Journal of Physics, 14(1), 13063 (2012)}. 

\bibitem{tsc1} L.-P. Garcia-Pintos, N. Linden, A. S.L. Malabarba, A. J. Short and A. Winter, {\it Equilibration time scales of physically relevant observables}, \href{https://arxiv.org/abs/1509.05732}{arXiv:1509.05732}.

\bibitem{tsc5} P. Reimann, {\it Typical fast thermalization processes in closed many-body systems}, \href{https://doi.org/10.1038/ncomms10821}{Nature Communications, 7, 10821 (2016)}. 


\bibitem{tsc3} H. Wilming, M. Goihl, C. Krumnow and J. Eisert, {\it Towards local equilibration in closed interacting quantum many-body systems}, \href{https://arxiv.org/abs/1704.06291}{arxiv:1704.06291}.

\bibitem{tsc2} Th. R. de Oliveira, C. Charalambous, D. Jonathan, M. Lewenstein and A. Riera, {\it Equilibration time scales in closed many-body quantum systems}, \href{https://arxiv.org/abs/1704.06646}{arXiv:1704.06646}.


\bibitem{Shnir} A. I. Shnirelman, {\it Ergodic properties of eigenfunctions}, \href{http://mi.mathnet.ru/rus/umn/v29/i6/p181 }{Usp. Mat. Nauk {\bf 29}, 181 (1974)}.

\bibitem{Berry} M. V. Berry, {\it Regular and irregular semiclassical wavefunctions}, \href{https://doi.org/10.1088/0305-4470/10/12/016}{J. Phys. A {\bf 10},  2083 (1977)}


\bibitem{ETH2} J. Deutsch, {\it Quantum statistical mechanics in a closed system}, \href{https://doi.org/10.1103/PhysRevA.43.2046}{Phys. Rev. A {\bf 43}, 2046 (1991)}

\bibitem{ETH3} M. Srednicki, {\it Chaos and Quantum Thermalization}, \href{https://doi.org/10.1103/PhysRevE.50.888}{Phys. Rev. E {\bf 50}, 888 (1994)}.

\bibitem{ETH6} M. Srednicki, {\it Thermal fluctuation in quantized chaotic systems}, \href{ 	https://doi.org/10.1088/0305-4470/29/4/003}{J. Phys. A {\bf 29} L75 (1996)}.

\bibitem{ETH7} M. Srednicki, {\it The approach to thermal equilibrium in quantized chaotic systems}, \href{https://doi.org/10.1088/0305-4470/32/7/007}{J. Phys. A {\bf 32} 1163 (1999)}. 

\bibitem{ETH1} M. Rigol, V. Dunjko, M. Olshanii, {\it Thermalization and its mechanism for generic isolated quantum systems}, \href{https://doi.org/10.1038/nature06838}{Nature {\bf 452}, 854-858 (2008)}.

\bibitem{ETH8} M. Rigol, {\it Breakdown of Thermalization in Finite One-Dimensional Systems}, \href{https://doi.org/10.1103/PhysRevLett.103.100403}{Phys. Rev. Lett. {\bf 103}, 100403 (2009)}.

\bibitem{ETH9} G. Biroli, C. Kollath, and A. M. L\"auchli,  {\it Effect of Rare Fluctuations on the Thermalization of Isolated Quantum Systems}, \href{https://doi.org/10.1103/PhysRevLett.105.250401}{Phys. Rev. Lett. {\bf 105}, 250401 (2010)}.

\bibitem{ETH10} T. N. Ikeda, Y. Watanabe, and M. Ueda, {\it Eigenstate randomization hypothesis: Why does the long-time average equal the microcanonical average?}, \href{https://doi.org/10.1103/PhysRevE.84.021130}{Phys. Rev. E {\bf 84}, 021130 (2011)}.

\bibitem{ETH5} M. Rigol, M. Srednicki, {\it Alternatives to Eigenstate Thermalization}, \href{https://doi.org/10.1103/PhysRevLett.108.110601}{Phys. Rev. Lett.  {\bf 108} 110601 (2012)}.

\bibitem{ETH4} P. Reimann, {\it Eigenstate thermalization: Deutsch's approach and beyond}, \href{https://doi.org/10.1088/1367-2630/17/5/055025}{New J. Phys. {\bf 17}, 055025 (2015)}.

\bibitem{ETH0} L. D'Alessio, Y. Kafri, A. Polkovnikov, M. Rigol, {\it From Quantum Chaos and Eigenstate Thermalization to Statistical Mechanics and Thermodynamics}, \href{https://doi.org/10.1080/00018732.2016.1198134}{Adv. Phys. {\bf 65}, 239 (2016)}.

  



\bibitem[{De Palma} et~al.(2015){De Palma}, Serafini, Giovannetti, and
  Cramer]{DePalma2015}
G. {De Palma}, A. Serafini, V. Giovannetti and M. Cramer,
{\it Necessity of Eigenstate Thermalization},
\href{https://doi.org/10.1103/PhysRevLett.115.220401}{{Physical Review Letters}, {\bf 115}, 220401, (2015)}.


  
\bibitem{MaxShan} F. Anza, V. Vedral, {\it Information-theoretic equilibrium and observable thermalization}, \href{https://doi.org/10.1038/srep44066}{Sci. Rep. {\bf 7}, 44066 (2017)}.

\bibitem[Khemani et~al.(2014)Khemani, Chandran, Kim, and Sondhi]{Khemani2014}
V. Khemani, A. Chandran, H. Kim, and S~L Sondhi,
{\it Eigenstate thermalization and representative states on subsystems},
\href{https://doi.org/10.1103/PhysRevE.90.052133}{{Physical Review E}, {\bf 90}, 052133, (2014)}.

\bibitem{ETH24}  M.P. M\"uller, E. Adlam, Ll. Masanes and N. Wiebe, {\it Thermalization and canonical typicality in translation-invariant quantum lattice systems}, \href{https://doi.org/10.1007/s00220-015-2473-y}{Commun. Math. Phys. {\bf 340}, 499 (2015)}.

\bibitem[Kim et~al.(2014)Kim, Ikeda, and Huse]{Kim2014}
H. Kim, T.~N Ikeda, and D.~A Huse,
{\it Testing whether all eigenstates obey the eigenstate thermalization
  hypothesis},
\href{https://doi.org/10.1103/PhysRevE.90.052105}{{Physical Review E}, {\bf 90}, 052105, (2014)}.

\bibitem[Steinigeweg et~al.(2014)Steinigeweg, Khodja, Niemeyer, Gogolin, and
  Gemmer]{Steinigeweg2014}
R~Steinigeweg, A~Khodja, H~Niemeyer, C~Gogolin, and J~Gemmer,
{\it Pushing the limits of the eigenstate thermalization hypothesis
  towards mesoscopic quantum systems},
\href{https://doi.org/10.1103/PhysRevLett.112.130403}{{Physical Review Letters}, {\bf 112}, 1, (2014)}.

\bibitem[Beugeling et~al.(2014)Beugeling, Moessner, and Haque]{Beugeling2014a}
W~Beugeling, R~Moessner and M. Haque,
{\it Finite-size scaling of eigenstate thermalization},
\href{https://doi.org/10.1103/PhysRevE.89.042112}{{Physical Review E}, {\bf 89}, 042112, (2014)}.

\bibitem[Torres-Herrera and Santos(2014)]{Torres-Herrera2014}
E~J Torres-Herrera and Lea~F Santos,
{\it Local quenches with global effects in interacting quantum systems},
\href{https://doi.org/10.1103/PhysRevE.89.062110}{{Physical Review E}, {\bf 89}, 062110, (2014)}.

\bibitem[Khodja et~al.(2015)Khodja, Steinigeweg, and Gemmer]{Khodja2014}
A. Khodja, R. Steinigeweg and J. Gemmer,
{\it Relevance of the eigenstate thermalization hypothesis for thermal
  relaxation},
\href{https://doi.org/10.1103/PhysRevE.91.012120}{{Physical Review E}, {\bf 91}, 012120, (2015)}.

\bibitem[Beugeling et~al.(2015)Beugeling, Moessner, and Haque]{Beugeling2015}
W~Beugeling, R~Moessner and M. Haque,
{\it Off-diagonal matrix elements of local operators in many-body quantum systems},
\href{https://doi.org/10.1103/PhysRevE.91.012144}{{Physical Review E}, {\bf 91}, 012144, (2015)}.

\bibitem{smoothnessETH1} E.J. Torres-Herrera, D. Kollmar and L.F. Santos, {\it Relaxation and thermalization of isolated many-body quantum systems}. \href{https://doi.org/10.1088/0031-8949/2015/T165/014018}{Physica Scripta, T165, 14018.(2015)}.

\bibitem[Bartsch and Gemmer(2017)]{Bartsch2017}
C. Bartsch and J. Gemmer, {\it Necessity of eigenstate thermalisation for equilibration towards
  unique expectation values when starting from generic initial states}, \href{https://doi.org/10.1209/0295-5075/118/10006}{EPL (Europhysics Letters), 118, 10006, (2017)}.

\bibitem{WeakETH} T. Mori, {\it Weak eigenstate thermalization with large deviation bound}, ArXiv:1609.09776
\bibitem{ERH1} T. N. Ikeda, Y. Watanabe, and M. Ueda, {\it Eigenstate randomization hypothesis: Why does the long-time average equal the microcanonical average?}, \href{https://doi.org/10.1103/PhysRevE.84.021130}{Phys. Rev. E, 84(2), 21130 (2011)}. 

\bibitem{mondainirigol} R. Mondaini, M. Rigol, {\it Eigenstate thermalization in the two-dimensional transverse field Ising model: II. Off-diagonal matrix elements of observables}, arXiv: 1705.08058



\bibitem{MUB3} S. Bandyopadhyay, P.O. Boykin, V. Roychowdhury, F. Vatan, {\it A new proof for the existence of mutually unbiased bases}. ArXiv:quantum-ph/0103162v3

\bibitem{MUB4} J. Lawrence, C. Brukner, A. Zeilinger, {\it Mutually unbiased binary observable sets on N qubits}, \href{https://doi.org/10.1103/PhysRevA.65.032320}{Phys. Rev. A {\bf 65}, 032320 (2002)}.

\bibitem{MUB2} I. Bengtsson, {\it Three ways to look at Mutually Unbiased Basis}, \href{https://doi.org/10.1063/1.2713445}{AIP Conf. Proc. 889, 40 (2007)}.

\bibitem{MUB6} I. Bengtsson, W. Bruzda, A. Ericsson, J.-A. Larsson, W. Tadej, and K. Zyczkowski, {\it Hadamard matrices from mutually unbiased bases}, \href{http://dx.doi.org/10.1063/1.3456082}{J. Math. Phys. 48, 052106 (2007)}

\bibitem{MUB1} S. Wehner, A. Winter, {\it Entropic uncertainty relations - a survey}, \href{https://doi.org/10.1088/1367-2630/12/2/025009}{New J. Phys. {\bf 12} 025009 (2010)}.

\bibitem{MUB5} T. Durt, B. Englert, I. Bengtsson, K. Zyczkowski, {\it On mutually unbiased basis}, \href{https://doi.org/10.1142/S0219749910006502}{Int. J. Quantum Information, 8, 535-640 (2010)}.

\bibitem{EigEnt2} V. Alba, M. Fagotti, and P. Calabrese, {\it  Entanglement entropy of excited states}, \href{https://doi.org/10.1088/1742-5468/2009/10/P10020}{J. Stat. Mech. 0910:P10020, (2009)}.

\bibitem{EE} J. Eisert, M. Cramer and M.B. Plenio, {\it Colloquium: Area laws for the entanglement entropy} \href{https://doi.org/10.1103/RevModPhys.82.277}{Rev. Mod. Phys.{\bf 82}  277 (2010)}.
\bibitem{CFT1} F. Alcaraz, M. I. Berganza and G. Sierra, {\it Entanglement of Low-Energy Excitations in Conformal Field Theory}, \href{https://doi.org/10.1103/PhysRevLett.106.201601}{Phys. Rev. Lett. {\bf 106}, 201601 (2011)}.

\bibitem{EigEnt1} J. Bhattacharya, M. Nozaki, T. Takayanagi, and T. Ugajin, {\it Thermodynamical Property of Entanglement Entropy for Excited States}, \href{https://doi.org/10.1103/PhysRevLett.110.091602}{Phys. Rev. Lett. {\bf 110}, 091602 (2013)}.

\bibitem{CFT4} G. Wong, I. Klich, L. A. P. Zayas, and D. Vamana, {\it Entanglement temperature and entanglement entropy of excited states}, \href{https://doi.org/10.1007/JHEP12(2013)020}{JHEP 2013: 20. (2013)}.

\bibitem{EigEnt3} F. Ares, J. G. Esteve, F. Falceto and E. Sanchez-Burillo, {\it Excited state entanglement in homogeneous fermionic chains}, 	\href{https://doi.org/10.1088/1751-8113/47/24/245301}{J. Phys. A: Math. Theor. {\bf 47}, 245301 (2014)}.

\bibitem{CFT2} T. Palmai, {\it Excited state entanglement in one-dimensional quantum critical systems: Extensivity and the role of microscopic details}, \href{https://doi.org/10.1103/PhysRevB.90.161404}{Phys. Rev. B {\bf 90}, 161404(R) (2014)}.

\bibitem{EigEnt4} J. Multer, T. Barthel, U. Schollwock and V. Alba, {\it Bound states and entanglement in the excited states of quantum spin chains}, \href{https://doi.org/10.1088/1742-5468/2014/10/P10029}{J. Stat. Mech. P10029 (2014)}.


\bibitem{CFT3} D. E. Parker, R. Vasseur, J. E. Moore, {\it Entanglement Entropy in Excited States of the Quantum Lifshitz Model}, ArXiv:1702.07433 (2017)



\bibitem{geom} I. Bengtsson, K. Zyczkowski, {\it Geometry of quantum states}, Cambridge University Press, Cambridge, 2006

\bibitem{bloch} R. A. Bertlmann and P. Krammer, {\it Bloch vectors for qudits}, \href{https://doi.org/10.1088/1751-8113/41/23/235303}{J. Phys. A: Math. Theor. {\bf 41}, 235303 (2008)}.

\bibitem{poly} A. Bachem, W. Kern,  {\it Basic Facts in Polyhedral Theory}, \href{https://doi.org/10.1007/978-3-642-58152-6_7}{Linear Programming Duality, pp 113-142 (Springer Berlin Heidelberg) (1992)}.


 

\end{thebibliography}
\end{document}